\documentclass[conference]{IEEEtran}
\IEEEoverridecommandlockouts
\usepackage{cite}
\usepackage{amsmath,amssymb,amsfonts}
\usepackage{algorithmic}
\usepackage{graphicx}
\usepackage{textcomp}
\usepackage{xcolor}
\def\BibTeX{{\rm B\kern-.05em{\sc i\kern-.025em b}\kern-.08em
    T\kern-.1667em\lower.7ex\hbox{E}\kern-.125emX}}

\usepackage{amssymb}
\usepackage{amsmath,amsthm}
\usepackage{mathtools}
\usepackage[lined,boxed]{algorithm2e}

\usepackage{multirow}
\usepackage{tikz}
\usepackage{hyperref}
\usepackage{cleveref}

\DeclareMathOperator*{\argmin}{arg\,min}

\newcommand{\dotp}[2]{ \langle #1, #2 \rangle}
\newcommand{\norm}[1]{||#1||}

\newcommand{\lat}{\mathcal{L}}

\newcommand{\R}{\mathbb{R}}

\newcommand{\Z}{\mathbb{Z}}

\newcommand{\CVP}{\mathsf{CVP}}
\newcommand{\MDSP}{\mathsf{MDSP}}
\newcommand{\SVP}{\mathsf{SVP}}
\newcommand{\SVPS}{\mathsf{SVPS}}

\newcommand{\Id}{\mathbb{I}}

\newtheorem{theorem}{Theorem}
\newtheorem{corollary}[theorem]{Corollary}
\newtheorem{lemma}[theorem]{Lemma}
\newtheorem{definition}{Definition}

\newtheorem{claim}{Claim}

\begin{document}

\title{On the Maximum Distance Sublattice Problem and Closest Vector Problem
}

\author{\IEEEauthorblockN{Rajendra Kumar}
\IEEEauthorblockA{\textit{Department of Computer Science \& Engineering} \\
\textit{Indian Institute of Technology Delhi}\\
New Delhi, India \\
r\_kumar@cse.iitd.ac.in}
\and
\IEEEauthorblockN{Shashank K Mehta}
\IEEEauthorblockA{\textit{Department of Computer Science \& Engineering} \\
\textit{Indian Institute of Technology Kanpur}\\
Uttar Pradesh, India \\
skmehta@cse.iitk.ac.in}
\and
\IEEEauthorblockN{Mahesh Sreekumar Rajasree}
\hspace{50em}\IEEEauthorblockA{\textit{Department of Computer Science \& Engineering} \\
\textit{Indian Institute of Technology Delhi}\\
New Delhi, India \\
srmahesh1994@gmail.com}
}

\maketitle

\begin{abstract}
In this paper, we introduce the Maximum Distance Sublattice Problem ($\MDSP$). We observed that the problem of solving an instance of the Closest Vector Problem ($\CVP$) in a lattice $\lat$ is the same as solving an instance of $\MDSP$ in the dual lattice of $\lat$. We give an alternate reduction between the $\CVP$ and  $\MDSP$. This alternate reduction does not use the concept of dual lattice. 
\end{abstract}


\section{Introduction}
For any set of linearly independent vectors 
$B = \{\vec{b_1},\hdots,\vec{b_n}\}\in \mathbb{R}^{m\times n}$, a \emph{lattice} $\lat$ is defined to be the set of vectors that consists of the integer linear combinations of vectors from
$B$. Formally it is defined as follows.
\[\lat=\lat(\vec{b_1},\hdots,\vec{b_n})= \left\{\sum_{i=1}^n z_ib_i \mid z_1, \dots, z_n \in \Z \right\}\]
Here, we call $n$ the rank of the lattice $\lat$ and $m$ as the ambient dimension. We call the set $B$ a basis of the lattice. Note that, a lattice can have infinitely many bases. Lattices have an enormous number of applications in Number theory~\cite{lenstra1982factoring, kannan1987minkowski,frank1987application} and Cryptanalysis~\cite{shamir1982polynomial,brickell1984breaking}. In the last two decade lattices got special attention due to their applications in Cryptography. Lattice-based Cryptosystems are considered the most prominent candidate for Post-Quantum Cryptography~\cite{ajtai1996generating,micciancio2007worst,regev2009lattices,regev2006lattice}.

The Shortest Vector problem ($\SVP$) and Closest Vector problem ($\CVP$) are two
well known and widely studied lattice problems.
Given a basis $B$ of the lattice $\mathcal{L}$, the shortest vector problem is to find a shortest (in some norm, usually in Euclidean-norm) non-zero vector in the lattice. In the closest vector problem we are also given a target vector $\vec{t}$ in the vector space of the lattice and the goal is to find the lattice vector closest (usually in Euclidean-norm) to the target $\vec{t}$. $\CVP$ is known to be NP-hard for approximation factor less than $n^{1/\log\log n}$~\cite{arora1997hardness,cai1998approximating,dinur1998approximating}. $\SVP$ is only shown to be NP-hard to approximate with constant approximation factor only by a randomized reduction\footnote{It is an long standing open problem to show NP-hardness for SVP via a deterministic reduction.}~\cite{Ajtai1998TheSV,micciancio2001shortest,micciancio2013deterministic}. It is also known to be poly-time hard for approximation factor
$n^{\mathcal{O}(1/\log\log n)}$ under some complexity theoretic assumption~\cite{khot2005hardness,haviv2007tensor}. Recently, there is also a series of works on the fine grained hardness of $\CVP$ \cite{BGSQuantitativeHardness17, ABGSFinegrainedHardness19, aggarwal2023we} and $\SVP$ \cite{ASGapETH18}. It is also know that $\CVP$ is at least as hard as  $\SVP$ as there is an approximation factor, rank and dimension preserving reduction from $\SVP$ to $\CVP$\cite{goldreich1999approximating}. 
 
All known algorithms for $\SVP$ and $\CVP$ require at least exponential time.  Kannan \cite{kannan1987minkowski} gave an enumeration based algorithm for $\CVP$ which takes $n^{\mathcal{O}(n)}$ time and polynomial space. There are also some improvements on running time of Kannan's algorithm \cite{hanrot2007improved,micciancio2014fast}. In 2001, Ajtai, Kumar and Sivakumar gave the first $2^{\mathcal{O}(n)}$ time and space sieving algorithm for $\SVP$~\cite{ajtai2001sieve} and $\CVP$~\cite{ajtai2002sampling}. There has been extensive works to improve the sieving algorithms for $\SVP$ and $\CVP$ \cite{nguyen2008sieve,arvind2008some,blomer2009sampling,pujol2009solving,micciancio2010faster,hanrot2011algorithms}. Fastest known classical algorithm for $\SVP$ and $\CVP$ takes $2^{n+o(n)}$ time and space, based on Discrete Gaussian Sampling \cite{aggarwal2015solvinga,aggarwal2015solvingb}. Recently Aggarwal, Chen, Kumar and Shen gave a faster quantum algorithm for $\SVP$ that requires $2^{0.835n+o(n)}$ time and exponential size QRAM and classical space~\cite{aggarwal2020improved}.

In 1982, Lenstra, Lenstra and Lovasz~\cite{lenstra1982factoring} gave a polynomial time algorithm (known as
LLL) for finding an exponential approximation of the shortest vector in the
lattices. The applications of LLL are found in factoring polynomials over rationals, finding linear Diophantine approximations, cryptanalysis of RSA and other
cryptosystems~\cite{coppersmith1996finding,shamir1982polynomial,coppersmith1996finding2}. Babai ~\cite{babai1986lovasz} gave a polynomial time algorithm, which uses
LLL, for approximating $\CVP$ with exponential approximation factor. Schnorr
has given improvements over the LLL algorithm \cite{schnorr1987hierarchy,schnorr1994lattice}.

\subsection{Our Contributions:}
In this paper, we introduce the Maximum Distance Sublattice Problem ($\MDSP$). Given a lattice vector $\vec{v}$, the goal is to find a sublattice of $n-1$ rank whose distance from the lattice vector $\vec{v}$ is maximum. We first observe that the $\MDSP$ problem reduces to the $\CVP$ on the dual lattice. The main technical contribution of our work is a reduction between the $\MDSP$ and $\CVP$ without using the notion of the dual lattice. The reduction employs novel geometric results that might be of independent importance. Our reduction preserves the dimension and rank of the lattice\footnote{We say a reduction is dimension-preserving and rank-preserving as long as the rank and dimension increases (or decreases) at most by 1.}.   

\begin{theorem}
\label{thm:main}
There exists a polynomial time rank-preserving dimension-preserving many-one (Karp) reduction between $\MDSP$ and $\CVP$.
\end{theorem}

The proof of the theorem is presented in \Cref{MDSPCVP}. We state our reduction for only for exact problem. It is easy to extend it for any approximation factor.

\subsection{Organisation:}
The rest of the paper is organised as follows. In section 2, we provide definitions and the trivial reduction between $\CVP$ and $\MDSP$. Section 3 contains our new reduction between $\CVP$ and $\MDSP$.

\section{Preliminaries}

In this paper $\mathbb{Z}$, $\mathbb{R}$ and $\mathbb{Q}$ will denote the sets of integers, reals and 
rationals respectively. Vectors will be denoted by small letters as in $\vec{v}$ and matrices
and basis sets will be denoted in capital letters. We will use $\Id_{n}$ to denote the $n \times n$ identity matrix. Let $B = \{\vec{b_1},\dots,\vec{b_k}\}$ be a set of vectors in $\mathbb{R}^n$. 
The subspace of $\mathbb{R}^n$ spanned by $B$ will be denoted by $span(B)$.

In this paper, we will work with vector space $V = \R^n$. For any vectors $\vec{u}, \vec{v} \in \R^n$, we use the notation $\dotp{\vec{u}}{\vec{v}}$ to denote the dot-product of the two vectors, i.e., $\dotp{\vec{u}}{\vec{v}} = \sum_{i=1}^n \vec{u}_i \vec{v}_i$ and $\norm{\vec{u}}$ denotes the $\ell_2$ norm of the $\vec{u}$, i.e., $\norm{\vec{u}} = (\sum_{i=1}^n \vec{u}_i^2)^{1/2}$. For a subspace 
$S\subseteq \mathbb{R}^n$, $S^{\bot} = \{\vec{x} \in \mathbb{R}^n | \dotp{\vec{x}}{\vec{y}} = 0,\ \forall \vec{y} \in S\}$  is also a subspace and it is called the {\em orthogonal subspace} of $S$. 

\begin{definition}[Lattice]
Given a set of linearly independent vectors $B =\{\vec{b_1},\dots,\vec{b_m}\}$ in a vector space $V$, the lattice spanned by $B$  is the set
$$\mathcal{L}(B)=\left\{\sum_{i=1}^m c_i \vec{b_i} \;|\;  c_i \in \mathbb{Z}  \textrm{ for all } 1\leq i \leq m\right\}$$
\end{definition}

In other words, a lattice is an integral span of $B$. The set $B$ is referred to as a {\em basis} of the lattice. The {\em rank} of the lattice is the number of linearly independent vectors in $B$ and the dimension of a lattice is the dimension of the ambient vector space containing the lattice. In this paper, we denote $B$ by a matrix where column vectors are the vectors of the generating set.
In this representation, the rank of a lattice is the same as the rank of the matrix $B$. Similar to a vector space, a lattice has infinitely many bases.  We will need the concept of unimodular matrices to characterize the bases of a given lattice.


\begin{definition}[Unimodular Matrix]
A matrix $U \in \mathbb{Z}^{n\times n}$ which has a determinant equal to $1$ or $-1$, is called a unimodular 
matrix. 
\end{definition}

Notice that the inverse and the transpose of a unimodular matrix are also unimodular. The following theorem states that two bases generate the same lattice if they are related by a unimodular matrix.

\begin{theorem}\label{PNt2} $B$ and $B'$ (in matrix form) are bases of the same rank-$n$ lattice ${\mathcal L}$ in $\mathbb{R}^n$ if and only if there exists an $n\times n$ unimodular matrix $U$ such that $B' = B U$.
\end{theorem}

An important concept in lattice theory is the dual of a lattice which is defined as follows.

\begin{definition}[Dual Lattice]
Let $\lat = \lat(B)$ be a lattice in $\R^{n}$. Then, the dual lattice of $\lat$, denoted by $\lat^{*}$ is
    \[\lat^{*} = \{\vec{v} \ | \ \forall \vec{u} \in \lat, \dotp{\vec{v}}{\vec{u}} \in \Z\}\]
\end{definition}
Let $B$ be an invertible matrix. Then, it can be easily shown that if $B$ is the basis of $\lat$, then $D = (B^{-1})^{T}$ is a basis for the dual lattice $\lat^{*}$. $D$ is called the dual basis of $B$. Observe that from the definition of dual basis, we have $D^T B = I$.

\begin{claim} \label{dual}
If $D$ is the dual basis of $B$, then for a basis $B' =BU$ where $U$ is a unimodular matrix, the dual basis is $D' = D(U^{-1})^T$.
\end{claim}

We will now proceeds to define certain computationally hard problems in lattice theory.

\begin{definition}[Shortest Vector Problem ($\SVP$)]
Given a basis $B=\{\vec{b}_1,\hdots,\vec{b}_n\}$, find a shortest non-zero vector $\vec{v}$ in the lattice 
$\mathcal{L}(B)$, i.e
    \[\vec{v} \in \argmin_{\vec{u} \in \mathcal{L}(B)\setminus \{0\}} \norm{\vec{u}}\]
\end{definition}


\begin{definition}[Closest Vector Problem ($\CVP$)]
Given a basis $B$ and a vector $\vec{t}$, find a vector $\vec{v}$ in the lattice 
$\mathcal{L}(B)$ which is closest from $\vec{t}$, i.e
    \[\vec{v} \in \argmin_{\vec{u} \in \mathcal{L}(B)}\norm{\vec{u} - \vec{t}}\]
\end{definition}

In this paper, we assume the vector $\vec{t}$ in $\CVP$ instance is linearly independent of basis $B$. In the case where $t$ is not independent, we can increase the dimension of the vector space and obtain linear independence as follows. We work with $B'$ and $t'$ such that 
\[\vec{b_i'} = 
\begin{bmatrix}
    \vec{b_i} \\
    0
\end{bmatrix},
\vec{t'} = 
\begin{bmatrix}
    \vec{t} \\
    1
\end{bmatrix}\]
Except for a constant factor, this one-dimensional increase has no effect on our/existing algorithms' running time.

\begin{definition}
Given a basis $B = \{\vec{b_1},\dots,\vec{b_k}\}$ of a subspace in $\mathbb{R}^n$,
the subspace $span(B)$ has an orthogonal basis $B^* = \{\vec{b_1^*},\dots,\vec{b_k^*}\}$ given by
$\vec{b_i}^* = \vec{b_i} - \sum_{j=1}^{i-1}\mu_{ij}\vec{b_j}^*$ where 
$\mu_{ij} = \dotp{\vec{b_i}}{\vec{b_j}^*} / (\vec{b_j}^*)^2$. This transformation of the basis is called 
{\em Gram Schmidt orthogonalization}.
\end{definition}

Using a Gram Schmidt orthogonalization of a basis of a subspace $S$, it is easy to compute the projection of a vector $\vec{v}$ onto the subspace $S$ as follows. Let $B = \{\vec{b_1},\dots,\vec{b_k}\}$ be a basis of a $k$-dimensional subspace of $\mathbb{R}^n$ and $\vec{v}$
be a vector in $\mathbb{R}^n$. The projection of $\vec{v}$ on the subspace $S= span(B)$
is its component in $S$. If $B^*$ is an orthogonal basis of $span(B)$ (such as
the one computed by Gram-Schmidt orthogonalization), then the projection of $\vec{v}$ on $S$
is 
\[proj_{S}(\vec{v}) = \sum_{i=1}^k \dfrac{\dotp{\vec{v}^T}{ \vec{b_i^*}}}{\dotp{\vec{b_i^*}}{\vec{b_i^*}}} \cdot \vec{b_i^*}\] 
The component of
$\vec{v}$ perpendicular to $S$ is $\vec{v} - proj_S(\vec{v})$. It is equal to the projection
of $\vec{v}$ on $S^{\bot}$, i.e., $proj_{S^{\bot}}(\vec{v}) = \vec{v} - proj_S(\vec{v})$.
The {\em distance} of the point $\vec{v}$ from the subspace $S$ is the length of this vector. So
\[dist(\vec{v},S) = \norm{\vec{v} - proj_S(\vec{v})} = \norm{proj_{S^{\bot}}(\vec{v})} \]


We now proceed to define \emph{Maximum Distance Sublattice Problem}.

\begin{definition}[Maximum Distance Sublattice Problem ($\MDSP$)]
Given a basis $[\vec{v} \mid B]=\{\vec{v},\vec{b_1},\dots,\vec{b_n}\}$ for an $n+1$ dimensional lattice $\lat$, find 
$B'=\{\vec{b'_1},\dots,\vec{b'_n}\}$ such that $\{\vec{v},\vec{b'_1},\dots,\vec{b'_n}\}$
is also a basis for $\lat$ and the distance $dist(\vec{v},span(B'))$ is maximum.
Here, we call $\vec{v}$ the fixed vector.
\end{definition}

The following theorem shows that a solution $B'$ to the $\MDSP$ can be achieved from $B$ by adding integral multiples
of $\vec{v}$ to vectors in $B$.

\begin{theorem}
\label{t1}
Let $[\vec{v} \mid B]$ be a basis of an $n+1$ dimensional lattice $\lat$ in $\mathbb{R}^{n+1}$. Then for any 
basis of the lattice of the form $[\vec{v} \mid B'']$, there exists integers $\alpha_1,\alpha_2,\dots,\alpha_n$ such that $[\vec{v} \mid B']$ is also a lattice basis
and $span(B') = span(B'')$ where
	\[B' = B + [\alpha_1\vec{v},\alpha_2\vec{v}, \dots, \alpha_n\vec{v} ] \]
\end{theorem}

We have included a proof of the above theorem in the \Cref{app:proof-of-t1} as we were unable to provide a reference for it.

The following theorem shows a trivial reduction between $\SVPS$ and $\MDSP$.
\begin{theorem} \label{trivial}
There exist polynomial time rank and dimension preserving many-one (Karp) reductions between $\CVP$ and $\MDSP$.
\end{theorem}

\begin{proof}
We will show that $\MDSP(c)$ is equivalent to $\CVP$ on basis $([\vec{d}_1, \dots, \vec{d}_n])$ and target $\vec{u}$ where $[\vec{u}, \vec{d}_1, \dots, \vec{d}_n]$ is the dual basis of $[\vec{v}, \vec{b}_1, \dots, \vec{b}_n]$. We will first show the reduction from $\MDSP$ to $\CVP$ and since all the computations in the reduction are invertible, the other direction is trivial.

Let the input to $\MDSP$ be $B = [\vec{v}, \vec{b}_1, \dots, \vec{b}_n]$ and its dual basis be $D= [\vec{u}, \vec{d}_1, \dots, \vec{d}_n]$. From \Cref{t1}, we know that a solution $B' = [\vec{v}, \vec{b}_1', \dots, \vec{b}_n']$ to $\MDSP$  can be written as $B ' = BU = [\vec{v}, \vec{b}_1 + \alpha_1\vec{v}, \dots, \vec{b}_n + \alpha_n\vec{v}]$, i.e.,
\[
U = \begin{bmatrix}
    	1 & & & \vec{\alpha}^T & & \\
        \begin{matrix}
        0 \\
        \vdots \\
        0
        \end{matrix} & & & \mbox{\Large $\Id_{n}$} & &
	\end{bmatrix}
\]
where $\vec{\alpha}^{T} = [\alpha_1, \dots, \alpha_n]$ is an integer vectors. From \Cref{dual}, we know that the dual basis $D'$ of $B'$ is $D(U^{-1})^{T}$ where
\[
(U^{-1})^T = \begin{bmatrix}
    	1 & 0 & \dots & 0 \\
        -\vec{\alpha} & & \mbox{\Large $\Id_{n}$} & 
	\end{bmatrix}
\]
Therefore, $D' = [\vec{u} - \sum_{i=1}^n \alpha_i \vec{d}_i, \vec{d}_1, \dots, \vec{d}_n]$. Also, from the definition of dual basis, we have $(D')^{T}B' = I$, therefore,
\begin{equation}
    \dotp{\vec{v}}{\left(\vec{u} - \sum_{i=1}^n \alpha_i \vec{d}_i \right)} = 1
\end{equation}
Using the fact that $\dotp{\vec{a}}{\vec{b}} = \norm{\vec{a}} \cdot \norm{\vec{b}} \cdot \cos(\theta)$ where $\theta$ is the angle between $\vec{a}$ and $\vec{b}$, we get
\begin{equation}
\label{equ:v-u-relation}
    \norm{\vec{v}} \cdot \cos(\theta) = \dfrac{1}{\norm{\vec{u} - \sum_{i=1}^n \alpha_i \vec{d}_i}}
\end{equation}
where $\theta$ is the angle between $\vec{v}$ and $\vec{u} - \sum \alpha_i \vec{d}_i$. Using the definition of dual basis, we know that $\vec{u} - \sum \alpha_i \vec{d}_i$ is perpendicular to all $\vec{b}_i'$ because $D'$ is the dual of $B'$. Therefore, $\vec{u} - \sum \alpha_i \vec{d}_i$ is perpendicular to $span(\vec{b}_1', \dots, \vec{b}_n')$. This implies that $90-\theta$ is the angle between $\vec{v}$ and $span(\vec{b}_1', \dots, \vec{b}_n')$. Hence, $\norm{\vec{v}} \cdot \sin(90-\theta)$ is the perpendicular distance between $\vec{v}$ and $span(\vec{b}_1', \dots, \vec{b}_n')$. 

Recall that $B'$ is the solution to the $\MDSP$ instance, which means that the perpendicular distance between $\vec{v}$ and $span(\vec{b}_1', \dots, \vec{b}_n')$ is maximized. In other words, $\norm{\vec{v}} \cdot \sin(90-\theta)$ is maximized. Therefore, $\norm{\vec{u} - \sum \alpha_i \vec{d}_i}$ is minimized due to \Cref{equ:v-u-relation}. But, this is essentially computing the shortest vector in the shifted lattice $\vec{u} + \lat(\vec{d_1}, \dots, \vec{d_{n}})$, which is exactly $\CVP$ with the basis $\{\vec{d_1}, \cdots, \vec{d_n}\}$ and target $\vec{u}$. 
\end{proof}


\section{New Reduction between $\MDSP$ and $\CVP$} \label{MDSPCVP}
In this section, we prove our main theorem, i.e., \Cref{thm:main} which is reduction between $\MDSP$ and $\CVP$ which does not utilize the concept of dual lattices. Let $[\vec{v} \;|\; B]$ be an input to the $\MDSP$.

Keeping Theorem \ref{t1} in consideration, the maximum distance sub-lattice problem can be stated as 
follows. Given an $(n+1)$-dimensional lattice with basis $\{\vec{v},\vec{b_1},\dots,\vec{b_n}\}$, compute an alternative basis $\{\vec{v},\vec{b_1}+j_1\vec{v},\dots,\vec{b_n}+j_n\vec{v}\}$
such that the distance of point $v$ from the subspace spanned by
$\{\vec{b_1}+j_1\vec{v},\dots,\vec{b_n}+j_n\vec{v}\}$ is maximum, where $j_i\in \mathbb{Z}$ for all $i \in [n]$.

Let $P_{x_1,\dots,x_n}$ denote the subspace spanned by the vectors $\vec{b_1}+x_1\vec{v},
\dots, \vec{b_n} + x_n\vec{v}$ for $(x_1,\dots,x_n)\in \mathbb{R}^n$. Following result determines
the distance of the point $\vec{v}$ from $P_{x_1,\dots,x_n}$ for the special case when
 $\{\vec{v},\vec{b_1}, \dots,\vec{b_n}\}$ is an orthonormal basis.

\begin{lemma}\label{MCl1} Let $\{\vec{v},\vec{b_1},\dots,\vec{b_n}\}$ be an orthonormal basis. Then
the distance of point $\vec{v}$ from $P_{x_1,\dots,x_n}$ is $1/\sqrt{1+\sum_{i=1}^n x_i^2}$
for any $(x_1,\dots,x_n) \in \mathbb{R}^n$.
\end{lemma}


\begin{proof}
Let $\sum_{i}c_i(\vec{b_i}+x_i\vec{v})$ be the projection of vector $\vec{v}$ on $P_{x_1,\dots,x_n}$. 
Then $\vec{w} = \sum_ic_i(\vec{b_i}+x_i\vec{v}) - \vec{v}$ is the perpendicular drop from point 
$\vec{v}$ to the plane. This implies that for all $i \in [n]$,
\begin{equation}
\dotp{\vec{w}}{(\vec{b_i}+x_i\vec{v})} = 0
\end{equation}
By expanding the $\vec{w}$ term and crucially using the fact that the vectors are orthonormal, we get
\begin{align*}
    &\dotp{\vec{w}}{(\vec{b_i}+x_i\vec{v})} \\
    &= \dotp{\sum_{j=1}^{n} c_j(\vec{b_j}+x_j\vec{v}) - \vec{v}}{(\vec{b_i}+x_i\vec{v})} \\
    &= \dotp{\sum_{j=1}^{n} c_j(\vec{b_j}+x_j\vec{v})}{(\vec{b_i}+x_i\vec{v})} - \dotp{\vec{v}}{(\vec{b_i}+x_i\vec{v})} \\
    &= \dotp{\sum_{j=1}^{n} c_j(\vec{b_j}+x_j\vec{v})}{(\vec{b_i}+x_i\vec{v})} - x_i \\
    &= \sum_{j=1}^{n} \dotp{ c_j(\vec{b_j}+x_j\vec{v})}{(\vec{b_i}+x_i\vec{v})} - x_i \\
    &= \sum_{j\neq i} \dotp{ c_j(\vec{b_j}+x_j\vec{v})}{(\vec{b_i}+x_i\vec{v})} + c_i (1 + x_i^2) - x_i \\
    &= \sum_{j\neq i} (c_j x_j x_i) + c_i (1 + x_i^2) - x_i \\
    &= c_i + x_i \cdot \left ( \sum_{j=1}^n (c_j x_j) - 1 \right)
\end{align*}
By equating the last equation to $0$, we get $c_i = -x_i t$ where $t = \sum_{j=1}^n c_jx_j - 1$. This gives us
\begin{align*}
    \vec{w} &= \sum_{i=1}^n c_i(\vec{b_i} + x_i\vec{v}) - \vec{v} \\
    &= \sum_{i=1}^n (-x_i t) \cdot (\vec{b_i} + x_i\vec{v}) - \vec{v} \\
    &= \left( \sum_{i=1}^n -x_i t \cdot \vec{b_i} \right ) + \left( \sum_{i=1}^n -x_i^2 t - 1 \right) \vec{v} \\
\end{align*}
The square of the distance of $\vec{v}$ from the plane $P_{x_1, \dots, x_n}$ is 
\begin{align*}
    \norm{\vec{w}}^2 &= \sum_{i=1}^n c_i^2 + (\sum_{i=1}^n c_ix_i-1)^2 \\
    &= \sum_{i=1}^nc_i^2 + t^2 \\
    &= t^2(1 + \sum_{i=1}^nx_i^2)
\end{align*}

We now focus on expressing $t$ in terms of $x_i$'s. We have 
\begin{align*}
    t &= \sum_{i=1}^n x_ic_i - 1 \\
    &= -t\sum_{i=1}^n x_i^2 - 1 \\
    \implies t &= -1/(1+\sum_{i=1}^n x_i^2)
\end{align*} 
Plugging this in the expression for $\norm{\vec{w}}^2$ we get $\vec{w}^2 = 1/(1+\sum_{i=1}^n x_i^2)$.
\end{proof}

The distance of a vector from a plane $P$ is equal to the length of the vector's projection on the orthogonal plane $P^{\bot}$ and projection is directly proportional to the length of the vector. Hence we have a trivial consequence.

\begin{corollary} \label{MCc1}
Let $\{\vec{v},\vec{b_1},\dots,\vec{b_n}\}$ be an orthogonal basis in which all but $\vec{v}$ are unit vectors. Then the distance of point $\vec{v}$ from $P_{x_1,\dots,x_n}$ is $\norm{\vec{v}}/\sqrt{1+\norm{\vec{v}}^2\sum_{i=1}^{n} x_i^2}$ for any $(x_1,\dots,x_n) \in \mathbb{R}^n$.
\end{corollary}

\begin{proof}
In this case $\vec{v}$ is no longer a unit vector. The basis of  $P_{x_1\ldots,x_n}$ is $\{\vec{b}_1+x_1\vec{v}, \vec{b}_2+x_2\vec{v},\dots\}$. It is same as $\{\vec{b}_1+x'_1\vec{u}, \vec{b}_2+x'_2\vec{u},\dots\}$ where the additive vector $\vec{u} = \vec{v}/\norm{\vec{v}}$ is a unit vector as required in \Cref{MCl1} and $x'_i = \norm{\vec{v}}x_i$. From the lemma, the distance of the point $\vec{u}$ from $P_{x_1',\ldots,x_n'}$, is $1/\sqrt{1+\sum_i(x'_i)^2} = 
1/\sqrt{1+\norm{\vec{v}}^2\sum_ix_i^2}$. Hence the distance from $\vec{v}$ is $\norm{\vec{v}}/\sqrt{1+\norm{\vec{v}}^2\sum_ix_i^2}$. 
\end{proof}


We will now focus on the general case in which the vectors $\vec{b_i}$ are not necessarily orthogonal to the vector $\vec{v}$. Let $\vec{b'_i} = \vec{b_i} - \gamma_i \vec{v}$ be perpendicular to $\vec{v}$ for each $i$, where  
$\gamma_i \in \mathbb{R}, \forall i$. So $\gamma_i = \dotp{\vec{b_i}}{\vec{v}}/\norm{\vec{v}}^2$ and 
 the plane spanned by $\{\vec{b'_1}, \dots, \vec{b'_n}\}$ is perpendicular to $\vec{v}$.
Note that $\gamma_i$ need not be an integer. Note that a lattice vector $\vec{b_i}+j_i.\vec{v}$ can now be represented as $\vec{b'_i} + (\gamma_i+j_i)\vec{v}$ in the new reference frame. 

Consider the plane $P_{x_1,\dots,x_n}$ which is spanned by $\vec{b_1} + x_1\vec{v},
\dots, \vec{b_n}+ x_n\vec{v}$. In the new basis, we have
\begin{equation}
\label{eq:P-in-terms-of-b'}
    P_{x_1, \dots, x_n} = span(\vec{b'_1} + (\gamma_1+x_1)\vec{v},
\dots, \vec{b'_n} + (\gamma_n+x_n)\vec{v})
\end{equation}

Let us now transform the basis, $\{\vec{b'_1},\dots,\vec{b'_n}\}$, of the $n$-dimensional subspace
into an orthonormal basis. Let $B'$ denote the matrix in which column vectors are $\vec{b'_1},\vec{b'_2},
\dots, \vec{b'_n}$. Let $L$ be a linear transformation such that the column vectors of $B'' = B' L$
form an orthonormal basis. Denote the column vectors of $B''$ by $\vec{b''_1},\dots,\vec{b''_n}$
which are unit vectors and mutually orthogonal. Therefore,
\begin{equation}
    \label{eq:b''}
    \vec{b''_i} = \sum_{k=1}^n L_{ki}\cdot \vec{b'_k}
\end{equation}
Note that the new basis $\{\vec{b''_1},\dots,\vec{b''_n}\}$ spans the same subspace which is spanned by 
$\vec{b'_1},\dots,\vec{b'_n}$. Now $\{\vec{v}, \vec{b''_1},\dots, \vec{b''_n} \}$ forms 
an orthogonal basis such that all but $\vec{v}$ are unit vectors. 

The plane $P_{x_1,\dots,x_n}$ is spanned by $\vec{b'_1} + (\gamma_1+x_1)\vec{v}, \dots, \vec{b'_n} 
+ (\gamma_n+x_n)\vec{v}$. We will now focus on expressing this plane in terms of the unit vectors $\{\vec{b''_i}\}$. If we extend a line parallel to $\vec{v}$ from the point $\vec{b''_i}$ (where $\vec{v}$ and $\vec{b''_i}$ are perpendicular to each other, for all $i \in [n]$), 
then it must intersect this plane at one point, say, $\vec{b''_i} + y_i\vec{v}$. Then the plane 
spanned by $\{\vec{b''_1} + y_1\vec{v},\dots, \vec{b''_n} + y_n\vec{v}\}$ is $P_{x_1,\dots,x_n}$ 
itself.

Using \Cref{eq:b''}, we have
\begin{align*}
    &\vec{b''_i} + y_i \vec{v} = \sum_{k=1}^n L_{ki}\cdot \vec{b'_k} + y_i \vec{v} \\
    &= \sum_{k=1}^n L_{ki}(\vec{b'_k} + (\gamma_k+x_k) \vec{v}) 
- \sum_{k=1}^n L_{ki}(\gamma_k+x_k) \vec{v} + y_i \vec{v}
\end{align*}
By the choice of $y_i$, $\vec{b''_i} +y_i \vec{v}$ belongs to $P_{x_1,\dots,x_n}$. From \Cref{eq:P-in-terms-of-b'}, we know that vector $\vec{b'_k}+ (\gamma_k+x_k) \vec{v}$ also belongs to the plane for each $k$. But, $v$ 
does not belong to the plane because it is linearly independent from the set of vector $\{\vec{b}_k\}$. Thus, from the linear independence, we can conclude that
$$-\sum_{k=1}^n L_{ki}(\gamma_k+x_k) \vec{v} + y_i \vec{v} = 0$$
This implies that
\begin{align*}
    y_i &= \sum_{k=1}^n L_{ki}(\gamma_k+x_k)\\
    \implies \vec{y} &= L^T\cdot\vec{\gamma} + L^T\cdot \vec{x}
\end{align*}

The plane $P_{x_1,\dots,x_n}$ is spanned by $\vec{b''_1} + y_1\vec{v},\dots, \vec{b''_n} + y_n\vec{v}$
where $\{\vec{b''_1},\dots,\vec{b''_n}\}$ is an orthonormal basis and $\vec{v}$ is perpendicular to
each vector of the set. From Corollary \ref{MCc1}, the square of the distance of $\vec{v}$ from the 
plane $P_{x_1,\dots,x_n}$ is $\norm{\vec{v}}^2/(1+ \norm{\vec{v}}^2\sum_iy_i^2)$. 

Recall that our goal is to find a sub-lattice plane $P_{j_1,\dots,j_n}$,  where $\vec{j} \in \mathbb{Z}^n$, such that the distance from $\vec{v}$ is maximized. Equivalently, we want to find a sub-lattice plane such that
$\sum _i y_i^2 = \norm{\vec{y}}^2$ is minimized, i.e., to minimize the length of the vector $\vec{y}$. Let $\vec{x} = \vec{j} \in \mathbb{Z}^n$, then corresponding $\vec{y} = L^T\cdot \vec{\gamma} 
+ L^T\cdot \vec{j}$. 

We now proceed to construct a $\CVP$ instance that will solve the $\MDSP$ instance. We start define a lattice ${\mathcal L}_1$ with basis $L^T$, i.e., the row vectors of $L$ form a basis of $\lat_1$. 
We denote the rows of $L$ by $\{\vec{r_1},\dots,\vec{r_n}\}$. Let $\vec{z} = -L^T\cdot \vec{\gamma} 
= -\sum_i\gamma_i\vec{r_i}$. Then the length of the vector $\vec{y}$ is equal to the distance between 
the fixed point $\vec{z}$ and the lattice point $\sum_i j_i\vec{r_i}$ of ${\mathcal L}_1$. Thus 
the problem reduces to finding a lattice point of ${\mathcal L}_1$ closest to the point $\vec{z}$. 
Therefore, we have reduced $\MDSP$ to an instance of $\CVP$ where $\{\vec{r}_1,\dots,\vec{r}_n\}$ is a lattice basis and 
$\vec{z}$ is the fixed point.

The following lemma summarises the computations needed to convert a $\MDSP$ instance to a $\CVP$ instance.

\begin{lemma} 
\label{MDSPtoCVP}
Given a basis of an $(n+1)$-dimensional lattice $\{v,b_1,\dots,b_n\}$ as
an instance of $\MDSP$. Let $\vec{b'_i} = \vec{b_i} - \gamma_i \vec{v}$ for all $1\leq i \leq n$ 
where $\gamma_i = \dotp{\vec{b_i}}{\vec{v}}/\norm{\vec{v}}^2$.
Let $L$ be a linear transformation such that $B'' = B'\cdot L$ is an orthonormal basis.
Equivalently $\{\vec{b''_1},\dots,\vec{b''_n}\}$ is an orthonormal basis where 
$\vec{b''_i} = \sum_{k} (L^T)_{ik} \vec{b'_k}$. Let $\vec{r_i}$ denote the $i$-th row of $L$. 
Then the sub-lattice plane $P_{j_1,\dots,j_n}$ has maximum distance from the point $\vec{v}$ if 
$\sum_ij_i\vec{r_i}$ is a closest lattice vector for the $\CVP$ instance
in which the lattice basis is $\{\vec{r}_1,\dots,\vec{r}_n\}$ and the fixed point is
$-L^T\cdot \vec{\gamma}$.
\end{lemma}

The entire transformation involves only invertible steps hence the converse of the above claim also holds.

\begin{lemma} 
\label{CVPtoMDSP} 
Let the basis $\{\vec{s_1},\dots,\vec{s_n}\}$ and the fixed point $\vec{t} \in \mathbb{R}^{n+1}$ be an instance of $\CVP$. Let $L$ be the matrix in which $i$-th row is $\vec{s_i}$ for all $1\leq i \leq n$. Let ${\bf \gamma} = -(L^T)^{-1}\cdot \vec{t}$. Pick an arbitrary orthonormal basis $\{\vec{e_0},\vec{e''_1}, \dots,\vec{e''_n}\}$ for $\mathbb{R}^{n+1}$. Let $B''$ be the matrix with column vectors $\vec{e''_1},\dots,\vec{e''_n}$. Let $B' = B''\cdot L^{-1}$. Let $\vec{e'_i}$ denote the $i$-th column of $B'$. Let $\vec{e_i} = \vec{e'_i} + \gamma_i\vec{e_0}$. If the $\MDSP$ instance $\{\vec{e_0},\vec{e_1},\dots, \vec{e_n}\}$ has an optimum solution sub-lattice plane formed by $\{\vec{e_1} + j_1\vec{e_0},\dots, \vec{e_n} + j_n\vec{e_0}\}$, then $\sum_ij_i\vec{s_i}$ is the solution of the given $\CVP$ instance.
\end{lemma}
Finally, \Cref{thm:main} is obtained by combining \Cref{MDSPtoCVP} and \Cref{CVPtoMDSP}.

\bibliography{reference}
\bibliographystyle{ieeetr}

\appendix

\section*{Proof of \Cref{t1}}
\label{app:proof-of-t1}
In this section, we provide a proof for \Cref{t1}.
\begin{proof}
Since $[\vec{v} \mid B'']$ and $[\vec{v} \mid B]$ generate the same lattice, there exists a unimodular matrix $U'$ (refer \Cref{PNt2}) such that

\[
	[ \vec{v} \mid B'' ] = [\vec{v} \mid B]  U'
\]
where
\[
U'=  \begin{bmatrix}
    	1 & \beta_1 & \beta_2 & \dots & \beta_{n-1} & \beta_n \\
        \begin{matrix}
        0 \\
        \vdots \\
        0
        \end{matrix} & & & \mbox{\Large $U$} & &
     \end{bmatrix}
\]
The determinant $det(U') = 1 \times det(U) = \pm 1$, so
$det(U) = \pm 1$. Observe that $U' \in \Z^{(n+1) \times (n+1)}$, so $U \in \Z^{n\times n}$ and
is unimodular. So $U^{-1}$ exists and it is also unimodular.
Let us denote $[\beta_1, \beta_2, \dots, \beta_n]$ by $\vec{\beta}^T$. Then
\begin{align*}
	&[ v \mid B'' ]
    \begin{bmatrix}
    	1 & 0 & 0 & \dots & 0 & 0 \\
        \begin{matrix}
        0 \\
        \vdots \\
        0
        \end{matrix} & & & \mbox{\Large $U^{-1}$} & &
	\end{bmatrix} \\
    &=  [v \mid B] 
    \begin{bmatrix}
    	1 & & & \vec{\beta}^T & & \\
        \begin{matrix}
        0 \\
        \vdots \\
        0
        \end{matrix} & & & \mbox{\Large $U$} & &
	\end{bmatrix} 
    \begin{bmatrix}
    	1 & 0 &  \dots & 0 \\
        \begin{matrix}
        0 \\
        \vdots \\
        0
        \end{matrix} & & \mbox{\Large $U^{-1}$} &
	\end{bmatrix} \\
    &= [\vec{v} \mid B] 
    \begin{bmatrix}
    	1 & & & \vec{\beta}^T  U^{-1} & & \\
        \begin{matrix}
        0 \\
        \vdots \\
        0
        \end{matrix} & & & \mbox{\Large $U U^{-1}$} & &
	\end{bmatrix} \\
    &= [\vec{v} | B] 
    \begin{bmatrix}
    	1 & & & \vec{\beta}^T U^{-1} & & \\
        \begin{matrix}
        0 \\
        \vdots \\
        0
        \end{matrix} & & & \mbox{\Large $\Id_n$} & &
	\end{bmatrix}\\
    &= [\vec{v} \mid B] + [\vec{0} \mid \alpha_1\vec{v},\dots,\alpha_n\vec{v}]
\end{align*}
where $\vec{\beta}^{T}  U^{-1} = (\alpha_1,\dots,\alpha_n)^T$.
The left-hand side in the above equation is equal to $[\vec{v} \mid B''U^{-1}]$. So 
$B''  U^{-1} = B + [\alpha_1\vec{v},\dots,\alpha_n\vec{v}]$. 

The matrix $U^{-1}$ is unimodular so $B''$ and $B'=B'' U^{-1}$ span the same sub-lattice and
 $B' = B + [\alpha \vec{v},\dots,\alpha_n \vec{v}]$.
\end{proof}

\end{document}